%% file: paper.tex
\documentclass[submission,copyright,creativecommons]{eptcs}
\providecommand{\event}{ESSS'14} 

\input{packages}

\newcommand{\setoffending}{\mathit{offending\_flows}}
\newcommand{\getoffending}{\mathit{get\_offending\_flows}}
\newcommand{\getIFS}{\mathit{getIFS}}
\newcommand{\getACS}{\mathit{getACS}}
\newcommand{\Tau}{\mathit{T}}
\newcommand{\BigO}{\mathrm{O}}

\title{Directed Security Policies:\\A Stateful Network Implementation}
\date{\today}
\author{Cornelius Diekmann\textsuperscript{\textdagger} \qquad\qquad Lars Hupel\textsuperscript{\textdaggerdbl} \qquad\qquad Georg Carle\textsuperscript{\textdagger}
\institute{Technische Universit\"{a}t M\"{u}nchen}
\institute{\textsuperscript{\textdagger}Chair for Network Architectures and Services \qquad\qquad \textsuperscript{\textdaggerdbl}Chair for Logic and Verification
}
\email{diekmann@net.in.tum.de \qquad\qquad lars.hupel@tum.de \qquad\qquad carle@in.tum.de}}
\def\titlerunning{Directed Security Policies}
\def\authorrunning{C. Diekmann, L. Hupel, G. Carle}

\newtheoremstyle{thm}{3pt}{3pt}{\itshape}{}{\bfseries}{}{.5em}{}
\newtheoremstyle{normal}{3pt}{3pt}{}{}{\bfseries}{}{.5em}{}

\theoremstyle{normal}

\newcounter{definitioncounter}
\newtheorem{definition}[definitioncounter]{Definition}

\newtheorem*{example}{Example}

\theoremstyle{thm}

\newcounter{lemmacounter}
\newtheorem{lemma}[lemmacounter]{Lemma}

\newcounter{corollarycounter}
\newtheorem{corollary}[corollarycounter]{Corollary}

\newcounter{theoremcounter}
\newtheorem{theorem}[theoremcounter]{Theorem}

\makeatletter
\g@addto@macro\th@thm{\thm@headpunct{}}
\g@addto@macro\th@normal{\thm@headpunct{}}
\makeatother

\begin{document}

\maketitle

\begin{abstract}
Large systems are commonly internetworked. 
A security policy describes the communication relationship between the networked entities. 
The security policy defines rules, for example that $A$ can connect to $B$, which results in a directed graph. 
However, this policy is often implemented in the network, for example by firewalls, such that $A$ can establish a connection to $B$ and all packets belonging to established connections are allowed. 
This stateful implementation is usually required for the network's functionality, but it introduces the backflow from $B$ to $A$, which might contradict the security policy. 
We derive compliance criteria for a policy and its stateful implementation. 
In particular, we provide a criterion to verify the lack of side effects in linear time. 
Algorithms to automatically construct a stateful implementation of security policy rules are presented, which narrows the gap between formalization and real-world implementation. 
The solution scales to large networks, which is confirmed by a large real-world case study. 
Its correctness is guaranteed by the Isabelle/HOL theorem prover. 
\end{abstract}

\section{Introduction}
Large systems with high requirements for security and reliability, such as SCADA or enterprise landscapes, no longer exist in isolation but are internetworked~\cite{hansen2012research}. 
Uncontrolled information leakage and access control violations may cause severe financial loss -- as demonstrated by Stuxnet -- and may even harm people if critical infrastructure is attacked. 
Hence, network security is crucial for system security. 

A central task of a network security policy is to define the network's desired connectivity structure and hence decrease its attack surface against access control breaches and information leakage.
A security policy defines, among others things, rules determining which host is allowed to communicate with which other hosts.
One of the most prominent security mechanisms to enforce a policy are network firewalls.
For adequate protection by a firewall, its rule set is critical~\cite{bishop2003computer,bartal1999firmato}.
For example, let $A$ and $B$ be sets of networked hosts identified by their IP addresses.
Let $A \rightarrow B$ denote a policy rule describing that $A$ is allowed to communicate with $B$.
Several solutions from the fields of formal testing~\cite{brucker2008modelfwisabelle} to formal verification~\cite{fireman2006} can guarantee that a firewall actually implements the policy $A \rightarrow B$.
However, to the best of our knowledge, one subtlety between firewall rules and policy rules remains unsolved: 
For different scenarios, there are diverging means with different protection for translating $A \rightarrow B$ to firewall rules. 

\paragraph{Scenario 1}
Let $A$ be a workstation in some local network and $B$ represent the hosts in the Internet.
The policy rule $A \rightarrow B$ can be justified as follows:
The workstation can access the Internet, but the hosts in the Internet cannot access the workstation, \ie the workstation is protected from attacks from the Internet.
This policy can be translated to \eg the Linux iptables firewall~\cite{iptables} as illustrated in Figure~\ref{tab:intro:statefuliptables}.
The first rule allows $A$ to establish a new connection to $B$.
The second rule allows any communication over established connections in both directions, a very common practice.
For example, $A$ can request a website and the answer is transmitted back to $A$ over the established connection.
Finally, the last rule drops all other packets.
In particular, no one can establish a connection to $A$; hence $A$ is protected from malicious accesses from the Internet.

\begin{figure*}[t]
\centering
\noindent\begin{minipage}{.7\linewidth}
\footnotesize
\texttt{iptables -A INPUT -s $A$ -d $B$ -m conntrack --ctstate NEW -j ACCEPT}\\
\texttt{iptables -A INPUT -m conntrack --ctstate ESTABLISHED -j ACCEPT}\\ 
\texttt{iptables -A INPUT -j DROP}
\end{minipage}%
  \vskip-7pt%
  \caption{Stateful implementation of $A \rightarrow B$ in Scenario~1}%
  \label{tab:intro:statefuliptables}%
\end{figure*}

\begin{figure*}[t]
\centering
\noindent\begin{minipage}{.7\linewidth}
\footnotesize
\texttt{iptables -A INPUT -s $A$ -d $B$ -j ACCEPT}\\
\texttt{iptables -A INPUT -j DROP}
\end{minipage}%
  \vskip-7pt%
  \caption{Stateless implementation of $A \rightarrow B$ in Scenario~2}%
  \label{tab:intro:statelessiptables}%
\end{figure*}

\paragraph{Scenario 2}
In a different scenario, the same policy rule $A \rightarrow B$ has to be translated to a completely different set of firewall rules.
Assume that $A$ is a smart meter recording electrical energy consumption data, which is in turn sent to the provider's billing gateway $B$.
There, smart meter records of many customers are collected.
That data must not flow back to any customer, as this could be a violation of other customers' privacy.
For example, under the assumption that $B$ sends packets back to $A$, a malicious customer could try to infer the energy consumption records of their neighbors with a timing attack.
In Germany, the requirement for unidirectional communication of smart meters is even standardized by a federal government agency~\cite{bsi2013smartmeter}.
The corresponding firewall rules for this scenario can be written down as shown in Figure~\ref{tab:intro:statelessiptables}.
The first rule allows packets from $A$ to $B$, whereas the second rule discards all other packets.
No connection state is established; hence no packets can be sent from $B$ to $A$.

\smallskip

These two firewall rule sets were created from the same security policy rule \mbox{$A \rightarrow B$}.
The first implementation of ``$\rightarrow$'' is ``can initiate connections to'', whereas the second implementation is ``can send packets to''.
The second implementation appears to be simpler and more secure, and the firewall rules are justifiable more easily by the policy.
However, this firewall configuration is undesirable in many scenarios as it might affect the desired functionality of the network.
For example, surfing the web is not possible as no responses (\ie websites) can be transferred back to the requesting host.

A decision must be made whether to implement a policy rule $A \rightarrow B$ in the stateful (Figure~\ref{tab:intro:statefuliptables}) or in the stateless fashion (Figure~\ref{tab:intro:statelessiptables}).
The stateful fashion bears the risk of undesired side effects by allowing packet flows that are opposite to the security policy rule.
In particular, this could introduce information leakage.
On the other hand, the stateless fashion might impair the network's functionality. 
Hence, stateful flows are preferable for network operation, but are undesirable with regard to security. 
In this paper, we tackle this problem by maximizing the number of policy rules that can be made stateful without introducing security issues.

We can see that even if a well-specified security policy exists, its implementation by a firewall configuration remains a manual and hence error-prone task.
A 2012 survey~\cite{sherry2012making} of 57 enterprise network administrators confirms that a ``majority of administrators stated misconfiguration as the most common cause of failure'' \cite{sherry2012making}.
A study~\cite{databreach2009src} conducted by Verizon from 2004 to 2009 and the United States Secret Service during 2008 and 2009 reveals that data leaks are often caused by configuration errors~\cite{databreach2009}.

In this paper, we answer the following questions:
\begin{itemize}
	\item What conditions can be checked to verify that a stateful policy implementation complies with the directed network security policy rules?
	\item When can a policy rule $A \rightarrow B$ be upgraded to allow a stateful connection between $A$ and $B$?
\end{itemize}

\noindent
Our results apply not only to firewalls but to any network security mechanisms that shape network connectivity.

The outline of this paper is as follows. 
Section~\ref{sec:example} presents a guiding example. 
Section~\ref{sec:model} formalizes the key concepts of directed policies, security requirements, and stateful policies. 
Section~\ref{sec:requiremensstateful} discusses the requirements for a stateful policy to comply with a directed policy. 
Section~\ref{sec:algorithm} presents an algorithm to automatically derive a stateful policy. 
Sections~\ref{sec:computationalcomplexity} and \ref{sec:casestudy} evaluate our work: 
Section~\ref{sec:computationalcomplexity} discusses the computational complexity of the algorithm, and Section~\ref{sec:casestudy} presents a large real-world case study.

\section{Example}
\label{sec:example}

\begin{figure*}[h]
  \centering
  \hspace*{\fill}%
  \begin{subfigure}[t]{0.48\textwidth}
       \centering
       \captionsetup{width=0.80\textwidth}
       \includegraphics[width=1.0\textwidth]{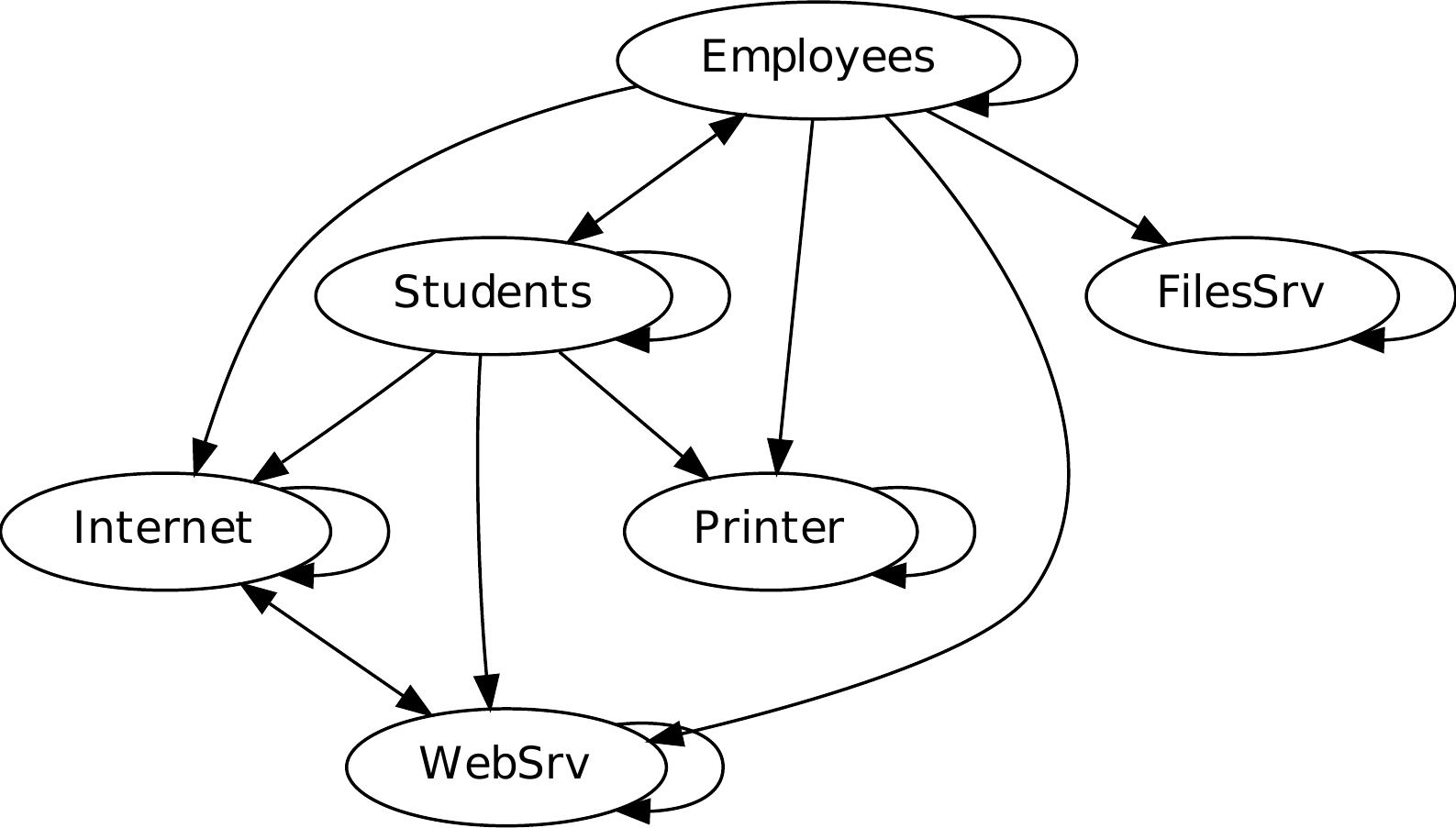}
  \caption{Network security policy}
  \label{fig:intro:policygraph}
  \end{subfigure}%
  \hspace*{\fill}%
  \begin{subfigure}[t]{0.48\textwidth}
       \centering
       \captionsetup{width=0.80\textwidth}
       \includegraphics[width=1.0\textwidth]{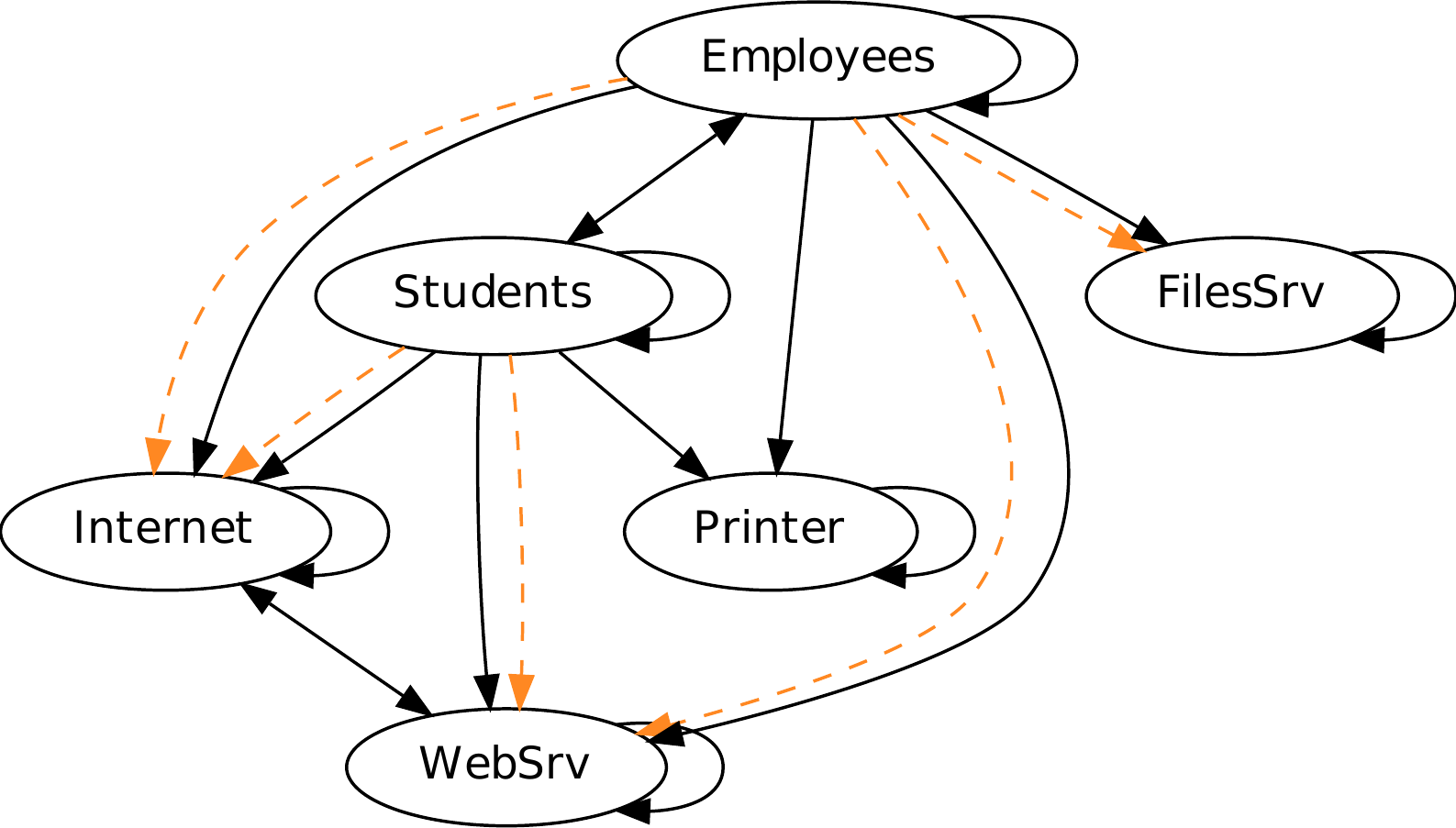}
  \caption{Stateful implementation}
  \label{fig:intro:statefulpolicygraph}
  \end{subfigure}%
  \hspace*{\fill}%
  \caption{The network security policy and its stateful implementation}
  \label{fig:intro}
\end{figure*}

\noindent
We introduce a network -- for a hypothetical university department -- to illustrate the problem with a complete example and outline the solution before we describe its formalization in the next section.

The network (depicted in Figure \ref{fig:intro}) consists of the following participants: the students, the employees, a printer, a file server, a web server, and the Internet.
The network security policy rules are depicted in Figure~\ref{fig:intro:policygraph} as a directed graph.
A security policy rule $A \rightarrow B$ is denoted by an edge from $A$ to $B$.
The security policy is designed to fulfill the following security invariants.

\begin{description}[\compact\setleftmargin{\parindent}] 
  \item[Access Control Invariants]
    The printer is only accessible by the employees and students; more formally, $\mathit{employees} \rightarrow \mathit{printer}$ and $\mathit{students} \rightarrow \mathit{printer}$. 
    The file server is only accessible by employees, formally \mbox{$\mathit{employees} \rightarrow \mathit{fileSrv}$}.
    The students and the employees are in a joint subnet that allows collaboration between them but protects against accesses from \eg the Internet or a compromised web or file server. 

  \item[Information Flow Invariants]
    The file server stores confidential data that must not leak to untrusted parties.
    Only the employees have the necessary security clearance to receive data from the file server.
    The employees are also trustworthy, \ie they may declassify and reveal any data received by the file server.
    The printer is an information sink.
    Confidential data (such as an exam) might be printed by an employee.
    No other network participants, in particular no students, are allowed to retrieve any information from the printer that might allow them to draw conclusions about the printed documents. 
    This can be formalized by ``$\mathit{*} \rightarrow \mathit{printer}$'' and ``$\mathit{printer} \nrightarrow \mathit{*}$''. 
\end{description}

\subsection*{Stateful Policy Implementation}
Considering Figure~\ref{fig:intro:policygraph}, it is desirable to allow stateful connections from the employees and students to the Internet and the web server.
Figure~\ref{fig:intro:statefulpolicygraph} depicts the stateful policy implementation, where the additional dashed edges represent flows that are allowed to be stateful, \ie answers in the opposite direction are allowed.
Only strict stateless unidirectional communication with the printer is necessary.
The students and employees can, as already defined by the policy, freely interact with each other. 
Hence stateful semantics are not necessary for these flows.

In this paper, we specify conditions to verify that the stateful policy implementation (\eg Figure~\ref{fig:intro:statefulpolicygraph}) complies with the directed security policy (\eg Figure~\ref{fig:intro:policygraph}).
We present an efficiently computable condition and formally prove that it implies several complex compliance conditions.
Finally, we present an algorithm that automatically computes a stateful policy from the directed policy and the security invariants.
We formally prove the algorithm's correctness and that it can always compute a maximal possible set of stateful flows with regard to access control and information flow security strategies.

\section{Formal Model}
\label{sec:model}
We implement our theory and formal proofs in the Isabelle/HOL theorem prover~\cite{isabelle2013}. 
It is based on a small inference kernel. 
All proof steps, done by either the user or by the (embedded or external) automated proof tactics and solvers, must pass this kernel. 
The correctness of Isabelle/HOL proofs therefore only depends on the correctness of the kernel. 
This architecture makes the system highly trustworthy, because the proof kernel consists only of little code, is widely used (and has been for over a decade) and is rigorously manually checked. 
In this paper, all proofs are verified by Isabelle/HOL. 
The corresponding theory files are publicly available (c.\hairspace{}f.\ Section~\nameref{sec:availability}).

The following notations are used in this paper. 
A total function from $\mathcal{A}$ to $\mathcal{B}$ is denoted by \mbox{$\mathcal{A} \Rightarrow \mathcal{B}$}. 
A logical implication is written with a long arrow ``$\Longrightarrow$''. 
Function application is written without parentheses: $f \ x \ y$ means ``$f$ applied to $x$ and $y$''. 
The set of Boolean values is denoted by the symbol $\mathbb B$. 

For readability, we only present the intuition behind proofs or even omit the proof completely. 
Whenever we omit a proof, we add an endnote that points to our formalization. 
We also add endnotes into the text which can be used to jump directly from a definition in this paper to the definition in the theory files. 
The endnotes are referenced by roman marks. 
For example, if the paper states ``note\endnotemark[4] that $A$ is equal to $B$'', then the corresponding formal, machine-verified proof can be found by following \endnotemark[4]. 

\paragraph*{Network Security Policy Rules}
We represent the network security policy's access rules as directed graph $G = (V,\, E)$. 
The type of all graphs is denoted by $\mathcal{G}$. 
For example, the policy that only consists of the rule that $A$ can send to $B$, denoted by $A \rightarrow B$, is represented by the graph $G = (\lbrace A, B \rbrace, \ \lbrace (A, B) \rbrace)$. 
An edge in the graph corresponds to a permitted flow in the network. 
We call this policy a \emph{directed policy}. 
In \S~\ref{subsec:statefulpolicy}, we will introduce the notion of a stateful policy. 

We consider only syntactically valid graphs. 
A graph is \emph{syntactically valid}\endnote{FiniteGraph.valid-graph} if all nodes in the edges are also listed in the set of vertices. 
In addition, since we represent finite networks, we require that $V$ is a finite set. 
This does not prevent creating nodes that represent collections of arbitrary many hosts, \eg the node $\mathit{Internet}$ in Figure~\ref{fig:intro:policygraph} represents arbitrarily many hosts.

\paragraph*{Network Security Invariants}
A security invariant $m$ specifies whether a given policy $G$ fulfills its security requirements. 
As we focus on the network security policy's access rules which specify which hosts are allowed to communicate with which other hosts, we do not take availability or resilience requirements into account. 
Instead, we deal with only the traditional security invariants that follow the principle ``prohibiting more is more or equally secure''. 
We call this principle \emph{monotonicity}. 
To allow arbitrary network security invariants, almost any total function $m$ of type $\mathcal{G} \Rightarrow \mathbb{B}$ can be used to specify a network security requirement. 

This model landscape is based on the formal model by Diekmann \cite{diekmann2014forte}. 
We distinguish between the two security strategies that $m$ is set to fulfill: 
\emph{Information flow security strategies} (IFS) prevent data leakage; 
\emph{Access control strategies} (ACS) are used to prevent illegal or unauthorized accesses. 

\begin{definition}[Security Invariant]
\label{def:securityinvariant}
A network security invariant $m$ is a total function \mbox{$\mathcal{G}  \Rightarrow\: \mathbb{B}$} with a security strategy (either IFS or ACS) satisfying the following conditions: 
\begin{itemize}
		\item {If no communication exists in the network, the security invariant must be fulfilled: $m \ (V,\ \emptyset) $}
		
		\item {Monotonicity: $m \ (V,E) \; \wedge \; E' \subseteq E \Longrightarrow m \ (V,E')$} 
	\end{itemize}
\end{definition}

\noindent
If there is a security violation for $m$ in $G$, there must be at least one set $F \subseteq E$ such that the security violation can be remedied by removing $F$ from $E$.\footnote{Since $m \ (V,\ \emptyset) $, it is obvious that such a set always exists.} 
We call $F$ offending flows. 
$F$ is \emph{minimal} if all flows $(s, r) \in F$ contribute to the security violation. 
For $m$, the set of all minimal offending flows can be defined. 
The definition $\setoffending$ describes a set of sets, containing all minimal candidates for $F$. 
\begin{IEEEeqnarray*}{l} \setoffending \ m \ G = 
\bigl\lbrace F \subseteq E \ \vert \ \neg\, m \ G \ \wedge \ m \ (V,\ E \setminus F)\ \ \wedge \ 
\forall (s,r) \in F.\ \neg\, m \ (V,\, (E \setminus F) \cup \lbrace (s,r) \rbrace)  \bigr\rbrace
\end{IEEEeqnarray*}

\noindent
The offending flows inherit $m$'s monotonicity property. 
The full proof can be found in our formalization.\endnote{offending-flows-union-mono} 
\begin{lemma}[Monotonicity of Offending Flows]
\label{lemma:mono-union-offending-flows}
\begin{IEEEeqnarray*}{l}
E' \subseteq E \Longrightarrow \bigcup \setoffending \ m \ (V, E') \subseteq \bigcup \setoffending \ m \ (V, E) 
\end{IEEEeqnarray*}
\end{lemma}
\noindent If there is an upper bound for the offending flows, it can be narrowed.\endnote{Un-set-offending-flows-bound-minus-subseteq} 
\begin{lemma}[Narrowed Upper Bound of Offending Flows]
\label{lemma-upperboundsubstract}
{Let $E'$ be a set of edges. 
If the offending flows are bounded, \ie if \ $\bigcup \setoffending \ m \ (V, E) \subseteq X$ holds, then
$\bigcup \setoffending \ m \ (V, E \setminus E') \subseteq X \setminus E'$.}
\end{lemma}
\begin{proof}
From Lemma~\ref{lemma:mono-union-offending-flows}, we have $\bigcup \setoffending \ m \ (V, E \setminus E') \subseteq \bigcup \setoffending \ m \ (V, E)$. 
This implies that $\left( \bigcup \setoffending \ m \ (V, E \setminus E') \right) \setminus E' \subseteq \left( \bigcup \setoffending \ m \ (V, E) \right) \setminus E'$.
Since the set of offending flows only returns subsets of the graph's edges, the left hand side can be simplified: 
$\bigcup \setoffending \ m \ (V, E \setminus E') \subseteq \left( \bigcup \setoffending \ m \ (V, E) \right) \setminus E'$. 
From the assumption, it follows that $\left( \bigcup \setoffending \ m \ (V, E) \right) \setminus E' \subseteq X \setminus E'$.
We finally obtain
\begin{IEEEeqnarray*}{l}
\bigcup \setoffending \ m \ (V, E \setminus E') \subseteq \left( \bigcup \setoffending \ m \ (V, E) \right) \setminus E' \subseteq X \setminus E'
\end{IEEEeqnarray*}
by transitivity.
\end{proof}

\begin{definition}[Security Invariants]
\label{def:securityinvarinatlist}
We call a finite list of security invariants $M = [m_1, m_2, ..., m_k]$ a network's security invariants. 
The functions $\getIFS \ M$ (and $\getACS \ M$) return all $m \in M$ with an IFS (and ACS, respectively) security strategy.
Additionally, we abbreviate all sets of offending flows for all security invariants with $\getoffending \ M \ G = \bigcup_{m \in M} \setoffending \ m \ G$. 
Similarly to $\setoffending$, it denotes a set of sets.
\end{definition}

\subsection{Stateful Policy Implementation}
\label{subsec:statefulpolicy}
We define a stateful policy similarly to a directed policy. 
\begin{definition}[Stateful Policy]
A stateful policy $\Tau = (V,\, E_\tau,\, E_\sigma)$ is a triple consisting of the networked hosts $V$, the flows $E_\tau$, and the stateful flows $E_\sigma \subseteq E_\tau$. 
\end{definition}

\noindent
The meaning of $E_\sigma$ is that these flows are allowed to be stateful. 
We consider the stateful flows $E_\sigma$ as ``upgraded'' flows, hence $E_\sigma \subseteq E_\tau$. 
This means that if $(s, r) \in E_\sigma$, flows in the opposite direction, \ie $(r, s)$ may exist. 
For a set of edges $X$, we define the \emph{backflows} of $X$ as $\overleftarrow{X} = \lbrace (r,s) \mid (s,r) \in X \rbrace$. 
Hence, the semantics of $E_\sigma$ can be described as that both the flows $E_\sigma$ and $\overleftarrow{E_\sigma}$ may exist. 
We define a mapping that translates a stateful policy $\Tau$ to a directed policy $G$ as 
$\alpha \ \Tau = (V,\, E_\tau \cup E_\sigma \cup \overleftarrow{E_\sigma})$. 

\begin{example}
The ultimate goal is to translate a directed policy $G = (V,\, E)$ to a stateful implementation $\Tau = (V,\, E_\tau,\, E_\sigma)$ that contains as many stateful flows $E_\sigma$ as possible without introducing security flaws. 
The trivial choice is $\Tau_{\mathrm{triv}} = (V,\, E,\, \emptyset)$. 
It fulfills all security invariants because $\alpha \ \Tau_{\mathrm{triv}} = G$. 
Since $E_\sigma = \emptyset$, it does not maximize the stateful flows. 
\end{example}

\noindent
Before discussing requirements for the compliance of $\Tau$ and $G$, we first have to define the requirements for a \emph{syntactically valid} stateful security policy.\endnote{valid-stateful-policy, stateful-policy-compliance} 
All nodes mentioned in $E_\tau$ and $E_\sigma$ must be listed in $V$. 
The flows $E_\tau$ must be allowed by the directed policy, hence $E_\tau \subseteq E$, which also implies $E_\sigma \subseteq E$ by transitivity. 
%
%
The nodes in $\Tau$ are equal to the nodes in $G$. 
This implies that $E_\tau$ and $E_\sigma$ are finite\endnote{valid-stateful-policy.finite-$\ast$}. 
In the rest of this paper, we always assume that $\Tau$ is syntactically valid. 

From these conditions, we conclude that $\Tau$ and $G$ are similar and $\Tau$ syntactically introduces neither new hosts nor flows. 
Semantically, however, $\alpha \ \Tau$ adds $\overleftarrow{E_\sigma}$, which might introduce new flows. 
Hence, the edges of $\alpha \ \Tau$ need not be a subset of $G$'s edges (nor vice versa).

\section{Requirements for Stateful Policy Implementation}
\label{sec:requiremensstateful}
We assume that $G$ is a \emph{valid policy}. 
In addition to being syntactically valid, that means that all security invariants must be fulfilled, \ie $ \forall m \in M. \ m \ G$. 
We derive requirements to verify that a stateful policy $\Tau$ is a proper stateful implementation of $G$ without introducing security flaws.  

\subsection{Requirements for Information Flow Security Compliance}
Information leakages are critical and can occur in subtle ways. 
For example, the widely used transport protocol TCP detects data loss by sending acknowledgment packages. 
If $A$ establishes a TCP connection to $B$, then even if $B$ sends no payload, arbitrary information can be transmitted to $A$, \eg via timing channels, TCP sequence numbers, or retransmits. 
Therefore, we treat information flow security requirements carefully: 
When considering backflows, all information flow security invariants must still be fulfilled. 
\begin{IEEEeqnarray}{l}
\label{eq-ifs:fulfilled}
 \forall m \in \getIFS \  M. \ \ m \ (\alpha \ \Tau)
\end{IEEEeqnarray}

\subsection{Requirements for Access Control Strategies}
In contrast, the requirements for access control invariants can be slightly relaxed: 
If $A$ accesses $B$, $A$ might expect an answer from $B$ for its request. 
If $B$'s answer is transmitted via the connection that $A$ established, $B$ does not access $A$ on its own initiative. 
Only the expected answer is transmitted back to $A$. 
If $A$'s software contains no vulnerability that $B$ could exploit with its answer, no access violation occurs.\footnote{Note that we make an important assumption here. This assumption is justified as we only work on the network level and do not consider the application level, which is also the correct abstraction for network administrators when configuring network security mechanisms. It also implies that, as always, vulnerable applications with access to the Internet can cause severe damage.} 
This behavior is widely deployed in many private and enterprise networks by the standard policy that internal hosts can access the Internet and receive replies, but the Internet cannot initiate connections to internal hosts. 

Therefore, we can formulate the requirement for ACS compliance. 
Access control violations caused by stateful backflows can be tolerated.
However, \emph{negative side effects} must not be introduced by permitting these backflows. 
First, we present an example of a negative side effect.
Second, we derive a requirement for verifying the lack of side effects. 

\begin{example}
We examine a building automation network. 
Let $B$ be the master controller, $A$ a door locking mechanism, and $C$ a log server that records who enters and who leaves the building. 
The controller $B$ decides when the door should be opened and what to log. 
%
The directed policy is described by $G = (\lbrace A, B, C \rbrace,\  \lbrace (B, A), (B, C)\rbrace)$. 
The only security invariant $m$ is that $A$ is not allowed to transitively access $C$. 
Let~$\rightarrow^*$ denote the transitive closure of $\rightarrow$. 
Then, $m$ prohibits $A \rightarrow^* C$, but it does not prohibit $C \rightarrow^* A$. 
In this scenario, that means that the physically accessible locking mechanism must not tamper with the integrity of the log server.

Setting $E_\sigma = \lbrace (B,A) \rbrace$ gives $\Tau = (\lbrace A, B, C \rbrace,\ \lbrace (B, A), (B, C)\rbrace,\ \lbrace (B,A) \rbrace)$, and hence $\alpha \ \Tau = (\lbrace A, B, C \rbrace,\linebreak \lbrace (B, A), (B, C), (A, B)\rbrace)$. 
This attempt results in a negative side effect. 
We compute the offending flows for $m$ of $\alpha\ \Tau$ as $\lbrace \lbrace(B, C)\rbrace,\  \lbrace(A, B)\rbrace \rbrace = \lbrace \lbrace(B, C)\rbrace,\  \overleftarrow{E_\sigma} \rbrace$.
Clearly, a violation occurs in $\overleftarrow{E_\sigma}$. 
Additionally, there is a side effect: the flow from $B$ to $C$ could now cause a violation. 
Applied to our scenario, this means that in case the locking mechanism sends forged data to the controller, that data could end up in the log.  
This is a negative side effect. 
Hence $(B,A)$ cannot securely be made stateful. 
For completeness, note that because $A$ is just a simple physical actor which only executes $B$'s commands, there is no need for bidirectional communication. 
On the other hand, $(B,C)$ can be made stateful without side effects. 
\end{example}
  
\noindent
We formalize the requirement of ``no negative side effects'' as follows: The violations caused by any subset of the backflows are at most these backflows themselves. 
\begin{IEEEeqnarray}{l}
\label{eq-acs:subsets}
\forall X \subseteq \overleftarrow{E_\sigma}.\ \forall F \in \getoffending \ (\getACS \ M) \ (V,\, E_\tau \cup E_\sigma \cup X). \ F \subseteq X
\end{IEEEeqnarray}
In particular, all offending access control violations are at most the stateful backflows. This is directly implied by the previous requirement by choosing $X$ to be $\overleftarrow{E_\sigma}$ (recall the definition of $\alpha$).  
\begin{IEEEeqnarray}{l}
\label{eq-acs:allset}
\bigcup \getoffending \ (\getACS \ M) \ (\alpha \ \Tau) \subseteq \overleftarrow{E_\sigma}
\end{IEEEeqnarray}
Also, considering all backflows individually, they cause no side effects, \ie the only violation added is the backflow itself. 
 \begin{IEEEeqnarray}{l}
\label{eq-acs:singletonset}
\forall (r, s) \in \overleftarrow{E_\sigma}. \bigcup \getoffending \ (\getACS \ M) \ (V,\, E_\tau \cup E_\sigma \cup \lbrace (r, s) \rbrace)  \subseteq \lbrace (r, s) \rbrace
\end{IEEEeqnarray} 
It is obvious that (\ref{eq-acs:subsets}) implies both (\ref{eq-acs:allset}) and (\ref{eq-acs:singletonset}).\endnote{stateful-policy-compliance.compliant-stateful-ACS-only-state-violations-union, stateful-policy-compliance.compliant-stateful-ACS-no-state-singleflow-side-effect} 
The condition of (\ref{eq-acs:subsets}) is imposed on all subsets, thus ruling out all possible undesired side effects. 

However, translating (\ref{eq-acs:subsets}) to executable code results in exponential runtime complexity, because it requires iterating over all subsets of $\overleftarrow{E_\sigma}$. 
This is infeasible for any large set of stateful flows. 
In this paper, we contribute a new formula\endnote{stateful-policy-compliance.compliant-stateful-ACS}, which implies (\ref{eq-acs:subsets}) and hence (\ref{eq-acs:allset}) and (\ref{eq-acs:singletonset}). 
It has a comparably low computational complexity and thus enables writing executable code for the automated verification of stateful and directed policies. 
\begin{IEEEeqnarray}{l}
\label{eq-acs:theallimplyACSformula}
\bigcup \getoffending \ (\getACS \ M) \ (\alpha \ \Tau ) \ \subseteq\ \overleftarrow{E_\sigma} \setminus E_\tau
\end{IEEEeqnarray}
Obviously, the runtime complexity of (\ref{eq-acs:theallimplyACSformula}) is significantly lower than (\ref{eq-acs:subsets}) (see \mbox{\S\hairspace\ref{sec:computationalcomplexity}}).
The formula also bears great resemblance to (\ref{eq-acs:allset}). 
We explain the intention of (\ref{eq-acs:theallimplyACSformula}) and prove that it implies (\ref{eq-acs:subsets}). 

Note that $\overleftarrow{E_\sigma} \setminus E_\tau = \overleftarrow{\lbrace(s, r) \in E_\sigma \ \vert \ (r, s) \notin E_\tau \rbrace}$ \endnote{backflows-filternew-flows-state}, which means that it represents the backflows of all flows that are not already in $E_\tau$. 
In other words, it represents only the newly added backflows. 
For example, consider the flows between students and employees in Figure~\ref{fig:intro:statefulpolicygraph}: 
no stateful flows are necessary as bidirectional flows are already allowed by the policy, and the newly added backflows are represented by the dashed edges. 
Therefore, (\ref{eq-acs:theallimplyACSformula}) requires that all introduced violations are only due to the newly added backflows. 
This requirement is sufficient to imply (\ref{eq-acs:subsets}).\endnote{stateful-policy-compliance.compliant-stateful-ACS-no-side-effects}

\begin{theorem}[Efficient ACS Compliance Criterion]
For ACS, verifying that all introduced violations are only due to the newly added backflows is sufficient to verify the lack of side effects. 
Formally, $(\ref{eq-acs:theallimplyACSformula})~\Longrightarrow~(\ref{eq-acs:subsets})$. 
\end{theorem}
\begin{proof}
We assume (\ref{eq-acs:theallimplyACSformula}) and show (\ref{eq-acs:subsets}) for an arbitrary but fixed $X \subseteq \overleftarrow{E_\sigma}$. 
We need to show that $\forall F \in \getoffending \ (\getACS \ M) \ (V,\, E_\tau \cup E_\sigma \cup X). \ F \subseteq X$. 
We split $\overleftarrow{E_\sigma}$ into $\overleftarrow{E_\sigma} \ \setminus \ E_\tau$ and $\overleftarrow{E_\sigma} \setminus ( \overleftarrow{E_\sigma} \ \setminus \ E_\tau )$. 
Likewise, we can split $X$ into $X_1 \subseteq \overleftarrow{E_\sigma} \ \setminus \ E_\tau$ and $X_2 \subseteq \overleftarrow{E_\sigma} \setminus ( \overleftarrow{E_\sigma} \ \setminus \ E_\tau )$. 
Hence, $X_2 \subseteq E_\tau$ and immediately $E_\tau \cup X_2 = E_\tau$. 
This simplifies the goal as $X_2$ disappears from the edges: %
\vskip-18pt
\begin{IEEEeqnarray*}{l}%
  \forall F \in \getoffending \ (\getACS \ M) \ (V,\, E_\tau \cup E_\sigma \cup X_1). \ F \subseteq X
\end{IEEEeqnarray*}%
\vskip-2pt%
\noindent We show an even stricter version of the goal since $X = X_1 \cup X_2$. %
\vskip-18pt%
\begin{IEEEeqnarray*}{l}%
  \forall F \in \getoffending \ (\getACS \ M) \ (V,\, E_\tau \cup E_\sigma \cup X_1). \ F \subseteq X_1
\end{IEEEeqnarray*}
\vskip-2pt%
\noindent This directly follows\endnote{stateful-policy-compliance.compliant-stateful-ACS-no-side-effects-filternew-helper} by using Lemma~\ref{lemma-upperboundsubstract} and subtracting $(\overleftarrow{E_\sigma} \ \setminus \ E_\tau) \setminus X_1$ from (\ref{eq-acs:theallimplyACSformula}). 
\end{proof}

\section{Automated Stateful Policy Construction}
\label{sec:algorithm}
In this section, we present algorithms to calculate a stateful implementation of a directed policy for a given set of security invariants using (\ref{eq-ifs:fulfilled}) and (\ref{eq-acs:theallimplyACSformula}). 

Instead of a set, the algorithms' last parameter is a list because the order of the elements matters. 
We denote the list cons operator by ``$::$''. 
For example, ``$e::\mathit{es}$'' is the list with the first element $e$ and a remainder list $\mathit{es}$. 
Since lists can be easily converted to finite sets, we make this conversion implicit for brevity. 
For example, for a list $a$, we will write the stateful policy as $(V,\, E,\, a)$, where $a$ is implicitly converted to a finite set. 

\subsection{Information Flow Security Strategies}
We start by presenting an algorithm which selects stateful edges in accordance to the IFS security invariants. 
The algorithm filters a given list of edges for edges which fulfill (\ref{eq-ifs:fulfilled}). 
It also takes as input the directed policy $G$, the security invariants $M$, and a list of edges as accumulator $a$. 
\begin{IEEEeqnarray*}{lcl}
\textnormal{\texttt{filterIFS}} \ G \ M \ a \;\textnormal{\texttt{[]}} & \ = \ \ & a\\
\textnormal{\texttt{filterIFS}} \ G \ M \ a \ (e ::\mathit{es}) & \ = \ \ & \textnormal{\texttt{if}} \;\;
  \forall m \in \getIFS \ \  M. \ m \ \left(\alpha \ \left(V, E, e::a \right)\right)
   \;\; \textnormal{\texttt{then}} \\
\ & & \quad \textnormal{\texttt{filterIFS}} \ G \ M \ (e :: a) \ \mathit{es}\\
\ & & \textnormal{\texttt{else}} \\
\ & & \quad \textnormal{\texttt{filterIFS}} \ G \ M \ a \ \mathit{es}
\end{IEEEeqnarray*}
The accumulator, initially empty, returns the result in the end. 
It is the current set of selected stateful flows. 
The algorithm is designed such that (\ref{eq-ifs:fulfilled}) always holds for $\Tau = (V, E, a)$. 
It simply iterates over all elements $e$ of the input list and checks whether the formula also holds if $e$ is added to $a$. 
If so, $e$ is added to the accumulator; otherwise, $a$ is left unchanged. 

Depending on the security invariants, multiple results are possible with this filtering criterion. 
The algorithm deterministically returns one solution. 
Users can influence the choice of edges that they want to be stateful by arranging the input list such that the preferred edges are listed first. 
If only one arbitrary solution is desired, lists and finite sets are interchangeable.

The algorithm is sound\endnote{filter-IFS-no-violations-correct} and complete.\endnote{filter-IFS-no-violations-maximal-allsubsets} 

\begin{lemma}[\texttt{filterIFS} Soundness]
If the directed policy $G = (V,\, E)$ is valid, 
then for any list $X \subseteq E$, the stateful policy $\Tau = (V,\, E,\, \textnormal{\texttt{filterIFS}}\ G \  M \;\textnormal{\texttt{[]}} \ X)$ fulfills (\ref{eq-ifs:fulfilled}). 
\end{lemma}


\begin{lemma}[\texttt{filterIFS} Completeness]
For $G = (V,\, E)$, 
let $E_\sigma = \textnormal{\texttt{filterIFS}} \ G \ M \ E$. 
Then, no non-empty subset can be added to $E_\sigma$ without violating (\ref{eq-ifs:fulfilled}). %
\begin{IEEEeqnarray*}{l}
\forall X \subseteq E \setminus E_\sigma, \  X \neq \emptyset. \ 
              \neg \forall m \in \getIFS \ M \ \left(\alpha \ (V,\, E,\, E_\sigma \cup X)\right)
\end{IEEEeqnarray*}
\end{lemma}

\subsection{Access Control Strategies}
The algorithm \texttt{filterACS} follows the same principles as \texttt{filterIFS}. 
%
\begin{IEEEeqnarray*}{lcl}
\textnormal{\texttt{filterACS}} \ G \ M \ a \;\textnormal{\texttt{[]}} & \ = \ \ & a\\
\textnormal{\texttt{filterACS}} \ G \ M \ a \ (e :: \mathit{es}) & \ = \ \ & \\
\IEEEeqnarraymulticol{3}{l}{
\qquad\quad 	\textnormal{\texttt{if}} \;\;
        e \notin \overleftarrow{E} \ \wedge \  \left(\forall F \in \getoffending \ (\getACS \ M) \ (\alpha (V,\, E,\, e :: a)). \ F \subseteq \overleftarrow{e :: a}\right)
	\;\; \textnormal{\texttt{then}}
}\\
\IEEEeqnarraymulticol{3}{l}{\qquad\quad \quad \textnormal{\texttt{filterACS}} \ G \ M \ (e :: a) \ \mathit{es}}\\
\IEEEeqnarraymulticol{3}{l}{\qquad\quad \textnormal{\texttt{else}}}\\
\IEEEeqnarraymulticol{3}{l}{\qquad\quad \quad \textnormal{\texttt{filterACS}} \ G \ M \ a \ \mathit{es}}
\end{IEEEeqnarray*}
As previously, the order of the elements in the list influences the choice of calculated stateful edges. 
Edges listed first are preferred. 
The algorithm is sound\endnote{filter-compliant-stateful-ACS-correct} and complete.\endnote{filter-compliant-stateful-ACS-maximal-allsubsets} 

\begin{lemma}[\texttt{filterACS} Soundness]
If the directed policy $G = (V,\, E)$ is valid, 
then for any list $X \subseteq E$, the stateful policy $\Tau = (V,\, E,\, \textnormal{\texttt{filterACS}}\ G \  M \;\textnormal{\texttt{[]}} \ X)$ fulfills (\ref{eq-acs:theallimplyACSformula}). 
\end{lemma}

\noindent
To show that \texttt{filterACS} computes a maximal solution, we must first identify the candidates that \texttt{filter\-ACS} might overlook. 
Flows that are already bidirectional need not be stateful. 
As illustrated in the example of Figure~\ref{fig:intro:statefulpolicygraph}, no added value is created if stateful connections between students and employees were allowed as no communication restrictions exist between these groups in the first place. 
Hence only $E \setminus \overleftarrow{E}$ is considered. 
\begin{lemma}[\texttt{filterACS} Completeness]
For $G = (V,\, E)$, 
let $E_\sigma = \textnormal{\texttt{filterACS}} \ G \ M \ E$. 
Then, no non-empty subsets $X \subseteq E \setminus ( E_\sigma \cup \overleftarrow{E})$ can be added to $E_\sigma$ without violating (\ref{eq-acs:theallimplyACSformula}). %
\begin{IEEEeqnarray*}{l}
\forall X \subseteq E \setminus ( E_\sigma \cup \overleftarrow{E} ), \  X \neq \emptyset. \ \ 
\neg \left( \bigcup \getoffending \ (\getACS \ M) \ (\alpha \ (V,\, E,\, E_\sigma \cup X) ) \ \subseteq \ \overleftarrow{E_\sigma \cup X} \setminus E\right)
\end{IEEEeqnarray*}
\end{lemma}

\subsection{IFS and ACS Combined}
Finally, we combine the previous section's algorithms to derive algorithms which compute a solution that satisfies all requirements of a stateful policy.

\noindent
The first algorithm\endnote{generate-valid-stateful-policy-IFSACS} simply chains \texttt{filterIFS} and \texttt{filterACS}. 
\begin{IEEEeqnarray*}{l}
\textnormal{\texttt{generate1}} \ G \ M \ e = \left(V,\, E,\, \textnormal{\texttt{filterACS}} \ G \ M \  (\textnormal{\texttt{filterIFS}} \ G \ M \ e)\right)
\end{IEEEeqnarray*}
The second algorithm\endnote{generate-valid-stateful-policy-IFSACS-2} takes the intersection of \texttt{filterIFS} and \texttt{filterACS}. 
\begin{IEEEeqnarray*}{l}
\textnormal{\texttt{generate2}} \ G \ M \ e = \left(V,\, E,\, \textnormal{\texttt{(filterACS}} \ G \ M \ e) \cap  \textnormal{\texttt{(filterIFS}} \ G \ M \ e) \right)
\end{IEEEeqnarray*}
\noindent
Both algorithms are sound.\endnote{generate-valid-stateful-policy-IFSACS-stateful-policy-compliance, generate-valid-stateful-policy-IFSACS-2-stateful-policy-compliance} 
It remains unclear whether both are equal in the general case. 
Furthermore, it is difficult to prove (or disprove) their completness, because both algorithms work on almost arbitrary functions $M$. 
However, we have formal proofs for the completeness of \texttt{filterACS} and \texttt{filterIFS} and the structure of \texttt{generate1} and \texttt{generate2} suggest completeness. 
In our experiments, \texttt{generate1} and \texttt{generate2} always calculated the same maximal solution. 

\begin{theorem}[\texttt{generate$\{$1,2$\}$} Soundness]
The algorithms \textnormal{\texttt{generate1}} and \textnormal{\texttt{generate2}} calculate a stateful policy that fulfills both IFS and ACS requirements. 
\end{theorem}

\begin{example}
Recall our running example. 
We illustrate how Figure~\ref{fig:intro:statefulpolicygraph} can be calculated from Figure~\ref{fig:intro:policygraph} and the security invariants. 
All ACS invariants impose only local---in contrast to transitive---access restrictions. 
Therefore, the ACS invariants lack side effects and \texttt{filterACS} selects all flows (excluding already bidirectional ones). 
The invariant that the file server stores confidential data also introduces no restrictions:  
Both $\mathit{filesSrv} \rightarrow \mathit{employees}$ and $\mathit{employees} \rightarrow \mathit{filesSrv}$ are allowed and since the employees are trusted, they can further distribute the data. 
Therefore, \texttt{filterIFS} applied on only this invariant correctly selects all flows. 
Up to this point, the network's functionality is maximized. 
However, since the printer is classified as information sink, it must not leak any data. 
Therefore, \texttt{filterIFS} applied to this invariant selects all but the flows to the printer. 
Ultimately, both \texttt{generate} algorithms compute\endnote{Impl\_\allowbreak{}List\_\allowbreak{}Playground\_\allowbreak{}ChairNetwork\_\allowbreak{}statefulpolicy\_\allowbreak{}example.thy} the same maximal stateful policy, illustrated in Figure~\ref{fig:intro:statefulpolicygraph}. 
The soundness and completeness of the running example is hence formally proven. 
The case study in Section~\ref{sec:casestudy} will focus on performance and feasibility in a large real-world example. 
\end{example}

\section{Computational Complexity}
\label{sec:computationalcomplexity}
The computational complexity of all presented formulae depends on the computational complexity of the security invariants $m \in M$. 
As we allow almost any function $m$ as security invariant, the computational complexity can be arbitrarily large. 
However, most of the security invariants we use in our daily business check a property over all flows in the network. 
Thus, the computational complexity of $m$ is linear in the number of edges, \ie $\BigO{}(\vert E \vert)$. 
The trivial computational complexity of $\setoffending$ is in \mbox{$\BigO{}(2^{\vert E \vert} \cdot \vert E \vert^2)$}, since it iterates over all subsets of $E$. 
However, given the structure of the security invariants we use, we provide proof\endnote{BLP-offending-set, CommunicationPartners-offending-set, ...} that the offending flows for our security invariants are uniquely defined \cite{diekmann2014forte}. 
They can be computed in \mbox{$\BigO{}(\vert E \vert)$}. 
The result is a singleton set whose inner set size is also in $\BigO{}(\vert E \vert)$. 
We present the computational complexity of our formulae and algorithms in this section for security invariants and offending flows with the mentioned complexity.\footnote{The computational complexity results are not formalized in Isabelle/HOL, because in its present state, there is no support for reasoning about asymptotic runtime behavior.} 
Our solution is not limited to these security invariants, but the computational complexity increases for more expensive security invariants.

We assume that set inclusion can be computed with the hedge union algorithm in $\BigO{}(k_i + k_j)$ for sets of size $k_{\{i,j\}}$. 
Since $E_\tau$ and $E_\sigma$ are bounded by $E$, set inclusion is in $\BigO{}(\vert E \vert)$.

Verifying information flow compliance, \ie (\ref{eq-ifs:fulfilled}), can be computed in $\BigO{}(\vert E \vert \cdot \vert M \vert)$. 
Hence, for a constant number of security invariants, the computational complexity is linear in the number of policy rules. 

To verify access control compliance, we first note that (\ref{eq-acs:subsets}) is in $\BigO{}(2^{\vert E \vert} \cdot \vert E \vert \cdot \vert M \vert)$ which is infeasible for a large policy. 
However, we provide (\ref{eq-acs:theallimplyACSformula}), which implies (\ref{eq-acs:subsets}), and can be computed in $\BigO{}(\vert E \vert \cdot \vert M \vert)$. 
Hence, for a constant number of security invariants, the computational complexity is linear in the number of policy rules. 

The \texttt{filter} and \texttt{generate$\{$1,2$\}$} algorithms only add $\BigO{}(\vert E \vert)$ to the complexity. 
Hence, for a constant number of security invariants, computing a stateful policy implementation from a directed policy is quadratic in the number of policy rules, which is feasible even for large policies with thousands of rules. 

\section{Case Study}
\label{sec:casestudy}
In a study, Wool~\cite{firwallerr2004} analyzed 37 firewall rule sets from telecommunications, financial, energy, media, automotive, and many other kinds of organization, collected in 2000 and 2001. 
The maximum observed rule set size was 2671, and the average rule set size was 144. 
Wool's study ``indicates that there are no good high-complexity rule sets''~\cite{firwallerr2004}. 
If in a scenario complicated rule sets are unavoidable, formal verification to assert their correctness is advisable. 

In this section, we analyze the firewall rule set of TUM's Chair for Network Architectures and Services. 
With a rule set size of approximately 2983 as of November 2013, this firewall configuration can be considered representatively large. 
Almost all rules are stateful, hence the firewall generally allows all established connections and only controls who is allowed to initiate a connection. 
We publish our complete data set, 
allowing others to reproduce our results and reuse the raw data for their research. 

\begin{figure}[!h]
  \centering
  		\includegraphics[width=\linewidth]{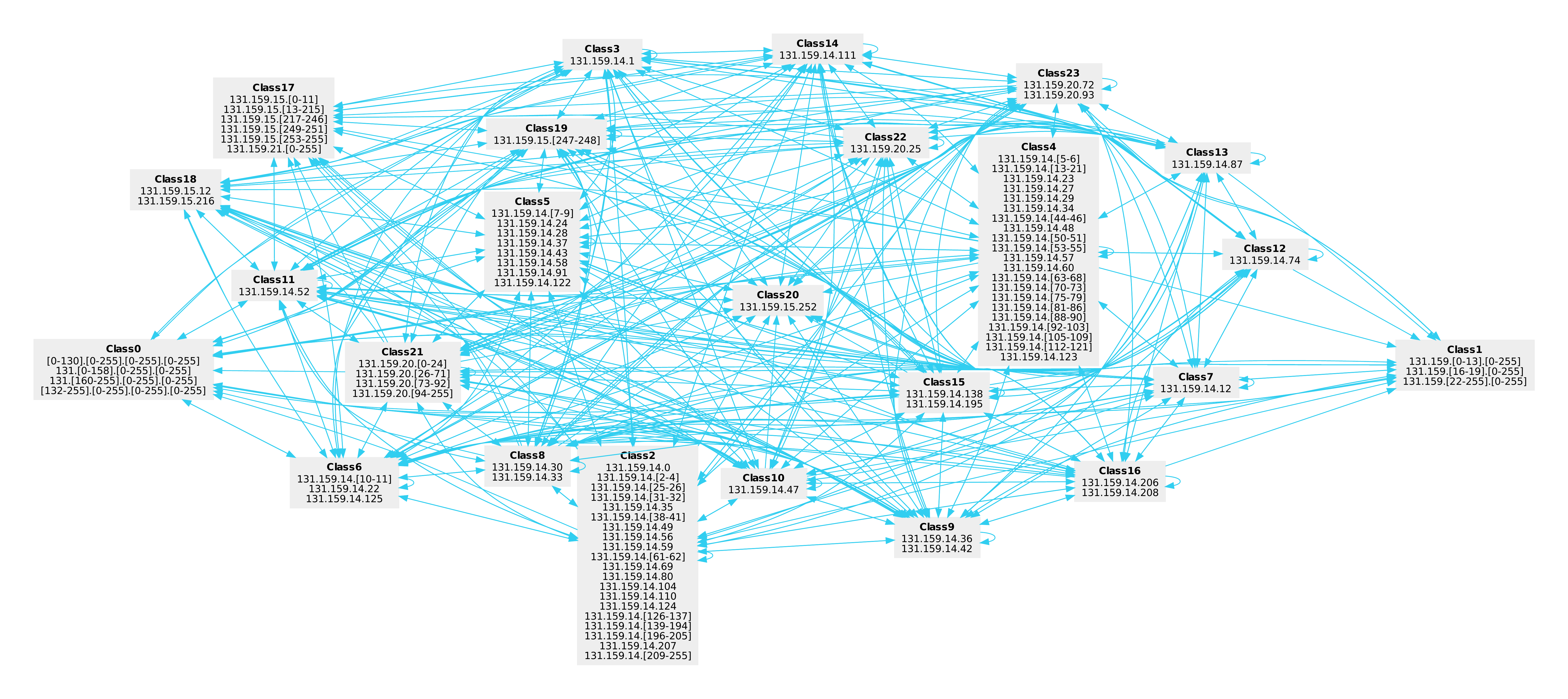}
  \caption{SSH landscape of the TUM Chair for Network Architectures and Services}
  \label{fig:eval:i8sshlandscape}
\end{figure}

\noindent
As there is no written formal security policy for our network, we reverse-engineered the security policy and invariants with the help of our system administrator. 
The firewall contains rules per IP range that permit the services which are accessible from some IP range. 
Most rules are similar to rule one in Figure~\ref{tab:intro:statefuliptables}. 
We regard the firewall rules about which hosts can initiate a connection as security policy. 
It is not unusual that the implementation is also the documentation~\cite[\S1]{cspfirewall}. 
We verify that the so derived security policy, \ie which hosts can initiate connections, corresponds to the stateful implementation, \ie all connections are stateful. 

In order to prepare the firewall rules as graph, we used \emph{ITval}~\cite{marmorstein2005itval} to first partition the IP space into classes with equivalent access rights~\cite{marmorstein2006firewall} which form the nodes of our policy. 
For each of these classes, we selected representatives and queried ITval for ``which hosts can this representative connect to'' and ``which hosts can connect to this representative''. 
This method is also suggested by Marmorstein \cite{marmorstein2006firewall}. 
The resulting IP ranges were mapped back to the classes. 
This generates the edges of the security policy graph. 
We asserted that these two queries result in the same graph. 
For brevity, we restrict our attention to the SSH landscape, \ie TCP port 22. 
The full data set is publicly available. 
The SSH landscape results in  a security policy with 24 nodes (sets of IP ranges with equal access rights) and 496 edges (permissions to establish SSH connections). 
The resulting graph is shown in Figure~\ref{fig:eval:i8sshlandscape}. 

A detailed discussion with our system administrator indicated that the graphical representation of the computed graph contains helpful information. 
It reveals that the computed policy does not exactly correspond to the firewall's configuration.
We could not clearly identify the cause for this discrepancy. 
However, the graph provides a sufficient approximation of our security policy. 
In the future, we will try to generate the graph using the approach by Tongaonkar, Niranjan, and Sekar \cite{tongaonkar2007inferring}, of which unfortunately no code is publicly available. 
In the long term, we see the need for formally verified means of translating network device configurations, such as firewall rule sets, SDN flow tables, routing tables, and vendor specific access control lists to formally accessible objects, such as graphs.

After having constructed the security policy, we implemented our security invariants. 
They state that our IP ranges form a big set of mostly collaborating hosts. 
As a general rule, internal hosts are protected from accesses from the outside world, 
but there are many exceptions. 

No IFS invariants exist and our ACS invariants cause no side effects. 
Note that we are evaluating neither the quality of our security policy nor the quality of our security invariants, but the quality of the stateful implementation in this large real-world scenario. 
As expected, our \texttt{generate$\{$1,2$\}$} algorithms identify all unidirectional flows as upgradable to stateful. 
This shows that the standard practice to declare (almost) all rules as stateful, combined with common simple invariants does not introduce security issues. 
For our invariants, our algorithms always generate a graph $\Tau$ such that $\alpha \ \Tau = (V,\ E \cup \overleftarrow{E})$. 
This means that in this scenario, we have a formal justification that all directed policy rules correspond to their stateful implementation, without any security concern. 
This maximizes the network's functionality without introducing security risks and is thus the optimal solution. 

This statement can be generalized to all networks without IFS invariants and without side effects in the ACS invariants. 
We provide formal proofs for both \texttt{generate$\{$1,2$\}$} algorithms.\endnote{generate-valid-stateful-policy-IFSACS-noIFS-noACSsideeffects-imp-fullgraph, generate-valid-stateful-policy-IFSACS-2-noIFS-noACSsideeffects-imp-fullgraph} 
Due to its simplicity, universality, and convenient implications for everyday use, we state this result explicitly. 

\begin{corollary}
If there are no information flow security invariants and all access control invariants of a directed policy lack side effects, a security policy can be smoothly implemented as stateful policy, without any security issues concerning state. 
\end{corollary}

\noindent
Our algorithms return this result, \ie $\alpha \ (\textnormal{\texttt{generate G M E}}) = (V,\ E \cup \overleftarrow{E})$. 
If there are information flow security invariants or access control invariants with side effects, our algorithms also handle these problems.

All results can be computed interactively on today's standard hardware. 
The graph preparation, which needs to be done only once, takes several seconds. 
Our \texttt{generate} algorithms take a few seconds. 
This shows the practical low computational complexity for a large real-world study.

\section{Related Work}
\label{sec:related}
In the research field of firewalls, several successful approaches to ease management~\cite{bartal1999firmato} and uncovering errors~\cite{fireman2006} exist. 
In \cite{cspfirewall}, the authors propose that a network security policy should exist in an informal language. 
A translation from the informal language to a formalized policy with an information content comparable to the directed policy in this work must be present. 
The same model for firewall rules and security policy is used. 
The authors model services, \ie ports, explicitly but ignore the direction of packets in their firewall model and are hence vulnerable to several attacks, such as spoofing~\cite{wool2004use}. 
Constraint Satisfaction Problem (CSP) solving techniques are used to test compliance of the security policy and the firewall rule set. 
Using Logic Programming with Priorities (LPP), Bandara \etal\cite{bandara2009using} build a framework to detect firewall anomalies and generate anomaly-free firewall configurations from a security policy. 
The authors explicitly point out the need for solving the stateful firewall problem.

Brucker~\etal\cite{brucker2008modelfwisabelle} provide a formalization of simple firewall policies in Isabelle/HOL and simplification rules for them. 
With this, they introduce HOL-TestGen/FW, a tool to generate test cases for conformance testing of a firewall rule set, \ie that the firewall under test implements its rule set correctly.  
In~\cite{brucker2013modelfwisabelle}, the authors augment their work with user-friendly high-level policies. 
This also allows the verification of a network specification with regard to these high-level policies. 

Guttman~\etal~\cite{guttman05rigorous,Guttman:1997:FilteringPostures} focus on distributed network security mechanisms, such as firewalls, filtering routers, and IPsec gateways. 
Security goals centered on the path of a packet through the network can be verified against the distributed network security mechanisms configuration. 

Using formal methods, network vulnerability analysis reasons about complete networks, including the services and client software running in the network. 
Using model checking~\cite{modelchecking2000} or logic programming~\cite{ou2005mulval}, network vulnerabilities can be discovered or the absence of vulnerabilities can be shown. 
One potential drawback of these methods is that the set of vulnerabilities must be known for the analysis, which can be an advantage for postmortem network intrusion analysis, but is also a downside when trying to estimate a network's future vulnerability.

Kazemian \etal\cite{kazemian2012HSA} present a method for the packet forwarding plane to identify problems such as reachability issues, forwarding loops, and traffic leakage. 
Considering the individual packet bits, the header space is represented by a $\langle\mathit{maximum\ packet\ header\ size\ in\ bits}\rangle$-dimensional space. 
An efficient algebra on the header space is provided which enables checking of the named use cases. 


\section{Conclusion}
\label{sec:conclusion}

Stateful firewall rules are commonly used to enforce network security policies. 
Due to these state-based rules, flows opposite to the security policy rules might be allowed. 
On the one hand, we argued that under presence of side effects or information flow invariants, a naive stateful implementation might break security invariants. 
On the other hand, declaring certain firewall rules to be stateless might impair the functionality of the network. 
This problem domain has often been overlooked in previous work. 

Verifying that a stateful firewall rule set is compliant with the security policy and its invariants is computationally expensive. 
In this work, we discovered a linear-time method and contribute algorithms for verifying and also for computing stateful rule sets. 
We demonstrated that these algorithms are fast enough for reasonably large networks, while provably maintaining soundness and completeness. 

Since the complete formalization, including algorithms and proofs, has been carried out in Isabelle/\allowbreak{}HOL, there is high confidence in their correctness. 
For the future, we see the need for verified translation methods from network device configurations to formally accessible objects, such as graphs.  

\section*{Acknowledgements \& Availability}
\label{sec:availability}
We thank our network administrator Andreas Korsten for his valuable input, his time and commitment. 
We appreciate Heiko Niedermayer's and Jasmin Blanchette's feedback.

This work has been supported by the German Federal Ministry of Education and Research (BMBF), EUREKA project SASER, grant 16BP12304, and by the European Commission, FP7 project EINS, grant 288021.

The Isabelle/HOL theory files can be obtained at \url{https://github.com/diekmann/topoS}. 
The complete raw data set of the firewall rules and a dump of our LDAP database, used to automatically construct some firewall rules, can be obtained at \url{https://github.com/diekmann/net-network}. 

	\bibliographystyle{eptcs}
	\bibliography{../literature}


\parindent 0pt
\begin{minipage}{\textwidth}
\begingroup
\def\enoteformat{\rightskip=0pt \leftskip=0pt \parindent=0pt \leavevmode{\makeenmark}\hspace*{3pt}}%
\def\enotesize{\footnotesize}
\begin{footnotesize} 
\theendnotes
\end{footnotesize}
\endgroup
\end{minipage}

\end{document}

%% file: packages.tex
\usepackage[pdftex]{graphicx}
\usepackage[utf8]{inputenc}

\usepackage[english]{babel}

\usepackage[cmex10]{amsmath}
\usepackage{amssymb}

\usepackage{amsthm}

\usepackage{url}

\newcommand{\hairspace}{\hspace{1pt}}
\newcommand{\eg}{\mbox{e.\hairspace{}g.} }  
\newcommand{\ie}{\mbox{i.\hairspace{}e.} }  

\newcommand{\etal}{\mbox{et~al.}\ }

	\usepackage{IEEEtrantools}

\usepackage{caption}
\usepackage{subcaption}

\usepackage[dvipsnames]{xcolor}

\usepackage{endnotescorny}

\def\notesname{Definitions, Lemmas and Theorems}
\renewcommand{\theendnote}{[{\roman{endnote}}]} 

\usepackage{booktabs}
\usepackage{array}
\usepackage{multirow}

\usepackage{pifont}
\usepackage{expdlist}

\hypersetup{
    colorlinks=true,
    linktocpage=true,
    breaklinks=true,
    bookmarksnumbered,
    bookmarksopen=true,
    bookmarksopenlevel=1,
    pdfhighlight=/O,
    pdfpagemode=UseOutlines,
    urlcolor=black,
    linkcolor=black,
    citecolor=black,
    pdfauthor={Cornelius Diekmann}
}